%% file: main.tex
\documentclass[english, 11pt]{article}
\usepackage{graphicx}

\usepackage{amssymb}
\usepackage{amstext}
\usepackage{amsthm}
\usepackage{float}
\usepackage{url}
\usepackage[ruled,linesnumbered]{algorithm2e}
\usepackage[noend]{algpseudocode}
\usepackage{amsmath}
\usepackage[utf8]{inputenc}
\usepackage[english]{babel}
\usepackage[margin=1in]{geometry}
\usepackage{color}
\usepackage{hyperref}
\usepackage{thm-restate}
\usepackage{complexity}
\newcommand{\dat}{\textsf{dat}}
\newcommand{\lap}{\mathcal{L}}
\newcommand{\mlap}{\mathcal{L}}
\newcommand{\apxlap}{\widetilde{\lap}}
\newcommand{\glap}{\lap_G}
\newcommand{\gadj}{\mathcal{A}_G}
\newcommand{\gdiag}{\mathcal{D}_G}
\newcommand{\ginci}{\mathcal{B}_G}
\newcommand{\Range}{\textsf{Range}}

\newcommand{\hlap}{\lap_H}
\newcommand{\hadj}{\mathcal{A}_H}
\newcommand{\hdiag}{\mathcal{D}_H}

\newcommand{\hnorm}{\mathcal{N}_H}
\newcommand{\pseudo}{+}
\newcommand{\x}{\textnormal{\textbf{x}}}
\newcommand{\y}{\textnormal{\textbf{y}}}

\newcommand{\z}{\textnormal{\textbf{z}}}
\newcommand{\w}{\textnormal{\textbf{w}}}

\newcommand{\bv}{{\textnormal{\textbf{b}}}}
\newcommand{\var}[1]{\mathbf{Var}\left[#1\right]}
\renewcommand{\R}{\mathbb{R}}
\newcommand{\otilde}{\widetilde{O}}
\newcommand{\vones}{\textbf{1}}
\newcommand{\vzeros}{\textbf{0}}

\newtheorem{theorem}{Theorem}
\newtheorem{lemma}{Lemma}
\newtheorem{definition}{Definition}
\begin{document}

\date{}

\title{Efficient $\otilde(n/\epsilon)$ Spectral Sketches \\for the Laplacian and its Pseudoinverse}

\author{Arun Jambulapati\footnote{This material is based upon work supported by the National Science Foundation Graduate Research Fellowship under Grant No. DGE-114747.}
	\\ Stanford University \\ \texttt{jmblpati@stanford.edu}
\and
Aaron Sidford \\
 Stanford University \\ \texttt{sidford@stanford.edu}}

\maketitle

\begin{abstract}

In this paper we consider the problem of efficiently computing $\epsilon$-sketches for the Laplacian and its pseudoinverse. Given a Laplacian and an error tolerance $\epsilon$, we seek to construct a function $f$ such that for any vector $x$ (chosen obliviously from $f$), with high probability $(1-\epsilon) x^\top A x \leq f(x) \leq (1 + \epsilon) x^\top A x$ where $A$ is either the Laplacian or its pseudoinverse. Our goal is to construct such a sketch $f$ 
efficiently and to store it in the least space possible. 

We provide nearly-linear time algorithms that, when given a Laplacian matrix $\mathcal{L} \in \mathbb{R}^{n \times n}$ and an error tolerance $\epsilon$, produce $\tilde{O}(n/\epsilon)$-size sketches of both $\mlap$ and its pseudoinverse. Our algorithms improve upon the previous best sketch size of $\otilde(n / \epsilon^{1.6})$ for sketching the Laplacian form by \cite{quadforms} and $O(n / \epsilon^2)$ for sketching the Laplacian pseudoinverse by \cite{BSS}.

Furthermore we show how to compute all-pairs effective resistances from our $\otilde(n/\epsilon)$ size sketch in $\otilde(n^2/\epsilon)$ time. This improves upon the previous best running time of $\otilde(n^2/\epsilon^2)$ by \cite{SS}. 

\end{abstract}

\thispagestyle{empty}\newpage
\clearpage \setcounter{page}{1}

\newpage

\input{intro.tex}
\input{prelim.tex}
\input{approach.tex}
\input{lsketch.tex}
\input{pseudosketch.tex}

\input{allpairs.tex}

\bibliography{main} 
\bibliographystyle{ieeetr}

\end{document}

%% file: intro.tex
\section{Introduction}
A Laplacian $\mlap \in \R^{n \times n}$ is a symmetric matrix, with non-positive off-diagonal entries, such that its diagonal entries are equal to the sum of its off-diagonal entries, i.e. $\mlap_{ij} = \mlap_{ji} \leq 0$ for all $i \neq j$ and $\mlap \vones = \vzeros$. An $\epsilon$-spectral sparsifier is a (hopefully sparse) Laplacian matrix $\apxlap \in \R^{n \times n}$ which preserves the quadratic form of $\mlap$ up to a multiplicative $(1\pm \epsilon)$ for all $\x\in \R^{n}$, that is 
\[
(1-\epsilon)\x^\top \lap \x \leq \x^\top \apxlap \x \leq (1+\epsilon) \x^\top \mlap \x ~.
\]
Since they were first proved to exist by Spielman and Teng in \cite{ST}, spectral sparsification has emerged as a powerful tool for designing fast graph algorithms and understanding the structure of undirected graphs \cite{Sherman13, KelnerLOS14, PengS14, AnariG15}. Spectral sparsifiers approximately preserve the eigenvalues of the original matrix, preserve all cuts in the graphs associated with the Laplacian matrices, serve as good preconditioners for solving Laplacian linear systems, and more. Furthermore, spectral sparsifiers preserve the quadratic form of the pseudoinverse and therefore preserve many natural measures of distance between vertices, like effective resistance distances and roundtrip commute times.

Given their numerous applications, obtaining faster algorithms for constructing sparser $\epsilon$-spectral sparsiers has been an incredibly active area of research over the past decade  \cite{SS,ST,BSS,ZhuLO15,LeeS15a,KoutisX16,LeeS17}. Recently, this work has culminated in the results of \cite{LeeS17} which showed how to construct $\epsilon$-spectral sparsifiers with $O(n/\epsilon^2)$-non-zero entries in nearly linear time. In addition, Andoni et al. \cite{quadforms} showed that for all $\epsilon > 1 / \sqrt{n}$, any $\epsilon$-spectral sparsifier (or even any data structure from which the quadratic form of any vector can be approximated) must have size at least $\Omega(n/\epsilon^2)$ thereby settling the size of spectral sparsifiers up to constant factors. 

A natural question remaining is whether relaxing the requirement of preserving $\x^\top \lap \x$ may allow for a smaller data structure. Instead of compressing a Laplacian so that we can approximate the quadratic for all vectors to $1 \pm \epsilon$-multiplicative accuracy we ask for the following weaker notion of approximately compressing a Laplacian, which is analagous to and inspired from the one in  \cite{quadforms}.\footnote{The definition in \cite{quadforms} is essentially the same as ours; the sole difference is that \cite{quadforms} allows the success probability to be a parameter.}
 
\begin{definition}[Spectral Sketch]
Given a symmetric PSD matrix $A \in \R^{n \times n}$ we say an algorithm produces an $s$-sparse $\epsilon$-spectral sketch of $A$ if it produces a function $f : \R^n \rightarrow \R$ such that $f$ can be stored and evaluated using only $O(s)$-space and such that for any particular $\x$ a query to $f$, i.e. $f(x)$, satisfies the following with probability at least $2/3$: 
\begin{equation}
\label{eq:spec_sketch_approx}
(1 - \epsilon) \x^\top A \x
\leq
f(\x) \leq (1 + \epsilon) \x^\top A \x ~.
\end{equation}
\end{definition}

Note that a spectral sketch of a Laplacian $\mlap$ can always be improved to have \eqref{eq:spec_sketch_approx} hold with probability $p$ simply by computing the sketch $\log(1/p)$ times independently of each other and outputting the median result. Consequently, an $s$-sparse $\epsilon$-spectral sketch of the Laplacian implies a $O(s \log n)$-sparse $\epsilon$-spectral sketch of the Laplacian where \eqref{eq:spec_sketch_approx} instead holds with high probability in $n$.

With this weaker notion of Laplacian approximation we can again ask the following question: how sparse can we make our sketch and how fast can we compute it? 

This notion of preserving the quadratic form was introduced and studied recently in \cite{quadforms} where they showed how to construct $\otilde(n/\epsilon^{1.6})$-sparse\footnote{Here and throughout the paper we use $\otilde(\cdot)$ notation to hide polylogarithmic factors.} $\epsilon$-spectral sketches, a much improved dependence on $\epsilon$. Furthermore, they showed that if queries $x$ to $f$ are restricted to be in $\{0,1\}^n$, i.e. cut queries, then $\otilde (n/\epsilon)$-sparse $\epsilon$-spectral sketches can be constructed, which is tight up to polylogarithmic factors. 


In this paper we show how to build upon the results in \cite{quadforms} to provide a nearly linear time algorithm to produce a general $\otilde(n/\epsilon)$-sparse $\epsilon$-spectral sketch of $\mlap$ with query time $\otilde(n/\epsilon)$. This improves on the sketch size to nearly the best size possible in light of the aforementioned lower bound in the restricted case of cut queries. Formally, in Section~\ref{sec:lap_sketch} we prove the following:

\begin{restatable}[Sketching the Laplacian]{theorem}{main}
\label{thm:main} For any polynomially-bounded weighted graph $G$ with Laplacian $\glap$, an $\otilde(n/\epsilon)$-sparse $\epsilon$-spectral sketch of $\glap$  can be constructed in $\otilde(m)$ time such that the sketch has query time $\otilde(n/\epsilon)$. 
\end{restatable} 

We also consider the problem of sketching the pseudoinverse of the Laplacian and show that we can produce an $\otilde(n/\epsilon)$-sparse $\epsilon$-spectral sketch of this matrix. In Section~\ref{sec:pseudo} we prove the following:

\begin{restatable}[Sketching the Pseudoinverse]{theorem}{pinv}
\label{thm:pinv} For any polynomially-bounded weighted graph $G$ with Laplacian $\glap$, it is possible to construct an $\otilde(n/\epsilon)$-sparse $\epsilon$-spectral sketch of $\glap^{\pseudo}$ in $\otilde(m)$ time such that the sketch has query time $\otilde(n/\epsilon)$.
\end{restatable}

This result is of intrinsic interest as it is the first sketch that improves upon the direct bounds for this problem achievable through Johnson-Lindenstrauss \cite{SS}, which yields $\otilde(n/\epsilon^2)$-sparse $\epsilon$-sketch sparsifiers for $\glap^\pseudo$ in $\otilde(m)$ time. This result is of intrinsic interest as the quadratic form of the pseudoinverse yields important graph quantities like effective resistances, which have seen many algorithmic uses over the years, from graph sparsification \cite{SS}, to random spanning tree sampling \cite{DurfeeKPRS17, KelnerM09, DingLP11}, to approximate maximum flow \cite{ChristianoKMST11}, among many others. 

As an interesting application of our sketch in time $\otilde (n^2/\epsilon)$ we show how to approximate the effective resistance between every pair of vertices in a graph. In Section~\ref{sec:all_pairs} we prove the following:

\begin{restatable}[All-Pairs Effective Resistances]{theorem}{aper}
\label{thm:aper} For any polynomially-bounded weighted graph $G$ with Laplacian $\glap$, we can generate a data structure in $\otilde(m)$ time that, with high probability, generates $1 \pm \epsilon$ approximations to the effective resistances between every pair of vertices in $\otilde(n^2/\epsilon)$ additional time. In addition, this data structure requires $\otilde(n/\epsilon)$ space to store.
\end{restatable}

This result improves upon the previous best running time of $\otilde(n^2/\epsilon^2)$ by Johnson-Lindenstrauss \cite{SS}, the running times of $O(n^\omega)$ for $\omega < 2.373$ \cite{LeGall14} achievable through fast matrix multiplication, and the running time of $\otilde(nm)$ achievable using nearly linear time Laplacian solvers  \cite{ST, CohenKPPR14} whenever the graph is sufficiently dense and $\epsilon$ is not too small. Moreover, this result improves upon the previous best space bound of $\otilde(n/\epsilon^2)$ by \cite{SS} for approximately storing all-pairs effective resistances. 

The rest of the paper is structured as follows. First, in Section~\ref{sec:prelim} we cover notation and preliminary facts we use throughout the paper. Then, in Section~\ref{sec:approach} we provide an overview of previous approaches and our approach to proving our main results. In Sections \ref{sec:lap_sketch}, \ref{sec:pseudo}, and \ref{sec:all_pairs} we prove our main results, Theorem~\ref{thm:main}, Theorem~\ref{thm:pinv}, and Theorem~\ref{thm:aper} respectively.

%% file: prelim.tex
 
\newpage
\section{Preliminaries}
\label{sec:prelim}

Throughout this paper we typically use $G$ to denote a  integer-weighted, undirected graph, with $n$ nodes, $m$ edges, and polynomially-bounded weights. For graph $G$ and vertex $u$, we let $\delta_u(G)$ be the \textit{weighted} degree of $u$: the sum of the weights of all edges incident upon $u$ in $G$. When the graph in question is clear, we  drop the $G$. We will also at times refer to \textit{unweighted} graphs, where all edge weights are assumed to be $1$. In this case we note that the weighted degree and the standard definition of degree are identical. Given $G = (V,E)$ and $S \subseteq V$ we define the subgraph of $G$ induced on $S$ as $H = (S,E')$, where $E'$ consists of the edges in $G$ with both endpoints in $S$. We call $H$ an induced subgraph of $G$ if there exists an $S$ where $G$ induced on $S$ is isomorphic to $H$.

For any weighted graph $H = (V, E, w)$ with non-negative integer edge weights $w \in \mathbb{Z}_{\geq 0}^{E}$ we let $\hlap \in \R^{n \times n}$ denote the \emph{Laplacian matrix} associated with $H$, i.e. for all $a,b \in V$ we have
\[
\mlap_H(a,b) = 
\begin{cases}
\delta_a & \text{if } a = b \\
-w_{\{a,b\}} & \text{if } \{a,b\} \in E \\
0 & \text{otherwise }
\end{cases}
~.
\]
We also let $\hdiag$ and $\hadj$ denote the degree and adjacency matrices of $H$, respectively, i.e. 
\[
\hdiag(a,b) = 
\begin{cases}
\delta_a & \text{if } a = b \\
0 & \text{otherwise }
\end{cases}
~
\text{ and }
~
\hadj(a,b) = 
\begin{cases}
w_{\{a,b\}} & \text{if } \{a,b\} \in E \\
0 & \text{otherwise }
\end{cases} ~.
\]
Note that $\hlap = \hdiag- \hadj$. 

For a symmetric matrix $A$ we use $A^{\pseudo}$ to denote the Moore–Penrose pseudoinverse for $A$ and for symmetric matrix $B$ we use $A \preceq B$ to denote the condition that $\x^\top A \x \leq \x^\top B \x$ for all $\x$; we define $\succeq$ analogously. We use $\vzeros$ and $\vones$ to denote the all zero and the all ones vector respectively when  the dimensions are clear from context. We use $\lambda_i(A)$ to denote the $i^{th}$ smallest eigenvalue of $A$.

Let $\hnorm  = \hdiag^{-1/2}\hlap\hdiag^{-1/2}$ be the normalized Laplacian of $H$.  The conductance of $H$ is
\[
\Phi_{H} = \min\limits_{S \subseteq V , S \notin \{\emptyset,V\}} \quad \frac{w(S, V - S)}{\min\{ \textsf{Vol}(S), \textsf{Vol}(V-S) \}}.
\]
where $w(S, V - S)$ is the sum of the weights of all edges with one endpoint in $S$ and on in $V - S$ and the volume $\textsf{Vol}(S)$ refers to the sum of the degrees of the vertices in $S$.

We use the discrete version of the following famous result relating conductance and eigenvalues:

\begin{lemma}[Cheeger's Inequality \cite{cheeger}]
\label{lem:cheeger}	
Let $H = (V,E,w)$ be an undirected weighted graph. Then 
\[
\lambda_2(\hnorm) \geq \frac{\Phi_{H}^2}{2}.
\]
\end{lemma}

Note that $\lambda_1(\hnorm) = \lambda_1(\hlap) = 0$, and $\lambda_2(\hlap)$ is nonzero if and only if $H$ is connected.

Finally, we use $\dat$ to denote the data used to construct a sketch. This data is a collection of objects of arbitrary types. We use the notation $a \in \dat$ to refer to objects $a$ being included in $\dat$ when the type of the objects are clear from context. 

Throughout the paper we use $\otilde(\cdot)$ to hide polylogarithmic factors. 

%% file: approach.tex
\section{Approach}
\label{sec:approach}

Here we summarize previous approaches to the problems of sketching quadratic forms, Laplacians, and their pseudoinverses as well as the problem computing all-pairs effective resistances (Section~\ref{sub:prev_approach}) and provide an overview of how improve upon these techniques to obtain our results (Section~\ref{sub:our_approach}).

\subsection{Previous Approaches}
\label{sub:prev_approach}

\textbf{Sketching Laplacians.} Despite the enormous amount of research on graph sparsifiers in recent years, the natural concept of $\epsilon$-spectral sketches was introduced only recently in \cite{quadforms}. As our work builds heavily on that of \cite{quadforms}, we first summarize and motivate their approach for obtaining $\otilde (n/\epsilon^{5/3})$-sparse $\epsilon$-spectral sketches.\footnote{The authors of \cite{quadforms} also provide a more complicated algorithm to construct $\otilde (n/\epsilon^{1.6})$-sparse $\epsilon$-spectral sketches, we describe the approach of their simpler $\otilde (n/\epsilon^{5/3})$-sparse sketch as it contains the main ideas we leverage. }

Given an undirected graph $G$ the authors of \cite{quadforms} first directly store $\gdiag$ in their sketch. This diagonal matrix spectrally dominates the adjacency $\gadj$ and only has a linear number of nonzero entries to store. To obtain greater accuracy in the sketch, they then sample from the adjacency matrix. However, doing so naively, e.g. by for example sampling random edges with uniform probabilities, does not yield the desired result; such a procedure would likely severely over or underestimate the number of edges incident upon low degree vertices. To fix this, \cite{quadforms} stores all of the edges incident upon vertices of low degree. Since these nodes have low degree, storing their edges does not take too much space. Sampling is then used only on the remaining portion of the graph. 

If the graph is has large conductance, i.e. is well-connected, to account for the contribution of edges between high degree nodes, they randomly sample a small number of edges incident upon each high degree vertex. With this information, given a vector $x$ the sketch in \cite{quadforms} simply computes $x^T\gdiag x - x^T \tilde{\mathcal{A}}_G x$, where  $\tilde{\mathcal{A}}_G$ is the adjacency matrix of the chosen edges. To see why this works, note that $x^T\gdiag x$ is effectively a weighted volume of $x$ and $x^T \gadj x$ is effectively a weighted sum of edges inside the set. Since they exactly computing the former and approximate the latter, if $G$ has high conductance then $x^T\glap x$ is not too small when compared to the error incurred when approximating $x^T \gadj x$ and this yields a spectral sketch as desired. 

To extend this result to general graphs, \cite{quadforms} recursively partitons the graph along sparse cuts, stores the crossing edges, and applies the above sketching procedure to each partition. Choosing parameters appropriately then yields their result.

\textbf{Sketching Laplacian Pseudoinverses.} In contrast to spectral sketching the Laplacian, the problem of $\epsilon$-spectral sketching the pseudoinverse of a Laplacian has remained unexplored.
Nevertheless, there are three previous approaches towards similar problems that we are aware of. 

First, it is known that the pseudoinverse of an $\epsilon$-spectral sparsifier of $\lap$ is an $\epsilon$-spectral sparsifier of $\lap^\pseudo$. Consequently, since $\otilde(n/\epsilon^2)$ bit $\epsilon$-spectral sparsifiers of Laplacians are known to exist and since spectral sparsification is harder than spectral sketching, $\otilde(n/\epsilon^2)$-sparse $\epsilon$-spectral sketches of the pseudoinverse are known. However, this is also known to be un-improvable through previously discussed lower bounds on spectral sparsification.

Another approach is to apply the Johnson-Lindenstrauss lemma to the the pseudoinverse. This gives an $\epsilon$-spectral sketch for $\lap^\pseudo$, but it too requires $\otilde(n/\epsilon^2)$ space to store. Further, this technique works for an arbitrary arbitrary positive semidefinite matrices as well for which it is known that $O(n / \epsilon^2)$-sparse $\epsilon$-spectral sketches do not exist  \cite{LN16}.

A third approach has been studied in the work on \emph{resistance sparsifiers} by \cite{DKW}. An $\epsilon$-resistance sparsifier of a given graph $G$ is a graph $H$ which retains the pairwise effective resistances of $G$ up to a $1 \pm \epsilon$ multiplicative factor. In \cite{DKW} it is shown how to construct  $\epsilon$-resistance sparsifiers of dense regular graphs with sufficiently high conductance with $\otilde(n /\epsilon)$ edges. In addition to requiring the input graph is dense regular and an expander, this approach is also limited in that the quadratic forms they are able to sketch are those corresponding to effective resistance queries, not the general quadratic forms that our definition of a spectral sketch requires. Despite these drawbacks, their algorithm does return a graph, 
while our algorithm merely provides a sketch. It is as of yet unclear if the work of \cite{DKW} could be extended in any way to the general graph and general query setting, and thereby allow us to compute graphs which sketch quadratic forms of Laplacians. 

\textbf{All Pairs Effective Resistances} As an application of the machinery we have developed, we consider the problem of approximating all pairs effective resistances, i.e. the effective resistance between every pair of vertices in the graph. 

The previous best known previous result for this problem of computing all pairs effective resistances was achieved by $\cite{SS}$. They achieved a running time of $\otilde (n^2/\epsilon^2)$, from a sketch of $\otilde (n/\epsilon^2)$ bits. This was done by noting that for any graph $G$, we can write $\glap = \ginci^\top \ginci$, where $\ginci$ is the incidence matrix of $G$, an explicit matrix with no more non-zero entries then $\glap$. With this, we have
\[R_{uv} = (\chi_u - \chi_v)^\top \glap^{\pseudo} (\chi_u - \chi_v) = (\chi_u - \chi_v)^\top \glap^{\pseudo} \ginci^\top \ginci \glap^{\pseudo} (\chi_u - \chi_v) = \|\ginci \glap^{\pseudo} (\chi_u - \chi_v)\|_2^2,
\]
where $\chi_i \in \mathbb{R}^n$ is $1$ in its $i^{th}$ component and $0$ everywhere else. Therefore, the effective resistances in a graph are the $\ell_2^2$ distances between vectors $\ginci \glap^{\pseudo} \chi_u$. With this observation, simply applying a Johnson-Lindenstrauss projection to each of these $n$ vectors approximately preserves the pairwise distances between them with high probability. Computing and storing these projected vectors allows for approximate recovery of all effective resistances in the claimed running time and space. 

\subsection{Our Approach}
\label{sub:our_approach}

To obtain our results, we build upon the foundations described in the previous subsection. We start by describing our improvement on \cite{quadforms} for sketching $\lap$. 

\textbf{Improvements for Sketching Laplacians.} Our approach builds on that of \cite{quadforms}. We note that the partitioning scheme used in \cite{quadforms} to generate subgraphs of high conductance is fundamentally limited in the degree of conductance ultimately attainable. Moreover, the partitioning procedure may often cut edges unnecessarily, as it will separate any sufficiently sparse cut, it is possible that this algorithm will break apart a subgraph with conductance only slightly below the threshold and needlessly store a large number of edges. Even worse, the subgraphs generated by this split may not necessarily have a better conductance than what was started with. As the number of crossing edges cannot be too large, this means that the best conductance guarantee a recursive scheme could hope to obtain is likely insufficient to obtain $\otilde(n/\epsilon)$-sparse $\epsilon$-spectral sketches. 

Consequently, we obtain our conductance guarantee in a different way. Our main insight is to use a different recursive partitioning procedure similar to many previous works on almost linear time graph algorithms \cite{ST,KelnerLOS14,MadryST15,CohenKPPRSV17}. Similar to these works, we use that the work of \cite{ST} shows that there exists a partition of the nodes of unweighted graphs into a logarithmic number of expander levels with ``good enough" expansion. More specifically, they show that any graph $G$ can be partitioned into $O(\log n)$ levels $G_1, G_2, ... G_k$, where each level's connected components are all expanders with  $\Omega(1/ \log n)$ conductance. Applying the above sketch to each of these connected components almost immediately gives us  $\otilde(n/\epsilon)$-sparse $\epsilon$-spectral sketches. 

Unfortunately, the approach used to find such a partition in \cite{ST} is to compute several exact sparsest cuts at every stage. As sparsest cut is known to be $NP$-hard, this procedure is infeasible. Mercifully, \cite{ST} also provides a nearly-linear time algorithm to compute an edge partition which, while not being a strict expander partition, is close enough to being an expander partition that we can prove that running \cite{quadforms}'s sketch on each of these connected components gives us our desired result anyways. This technique has found many uses in the previous literature and has given nearly-linear time algorithms for a diverse set of graph problems \cite{KelnerLOS14,MadryST15,CohenKPPRSV17}.

\textbf{Improvements for Sketching Laplacian Pseudoinverses.} To achieve our space improvement for sketching the Laplacian pseudoinverse, we provide a general reduction from the problem of sketching a matrix pseudoinverse $A^\pseudo$ to the problems of sketching $A$ and computing a (weaker) sparsifier of $A$. We construct a $\otilde(n/\epsilon)$-size spectral sketch of $\mlap^\pseudo$ by combining an $\otilde(n/\epsilon)$-size spectral sketch of $\mlap$ and a $(1 \pm  \sqrt{\epsilon})$-spectral sparsifier of $\mlap$. To query the sketch we solve a linear system in the computed sparsifier and apply our sketch for $\glap$ in a natural way to ``smooth" out the resulting estimate. This approach is inspired by fairly well known techniques in regression \cite{DWM} \cite{LS} which show that although and stochastic descent \cite{JulienSB12} that when sampling $k$ components of decomposable strongly-convex function solving on the samples gives a vector whose quality decays at a $1/k$-rate, rather than a $1/\sqrt{k}$ rate. While the mathematics behind this approach is known and fairly straightforward, its application for spectral sketches and effective resistance computation seem quite new as observed by the faster all-pairs effective resistance computation we discuss next. We hope that this idea may be useful in settings outside of the Laplacian paradigm. 

\textbf{Improvements to All Pairs Effective Resistances.} Finally, we describe our improvements to approximating all pairs effective resistances in a graph. We note that by our result on sketching $\lap^{\pseudo}$, we are able to store a sketch which can correctly approximate every $R_{uv} = (\chi_u - \chi_v)^\top \lap^\pseudo (\chi_u - \chi_v)$ with high probability. Our claimed space bound follows. However, this naive strategy does not provide the desired runtime improvement on computation. Since computing each $R_{uv}$ in this way requires approximately solving a linear system of $\otilde (n/\epsilon)$ nonzeroes computing all of the effective resistances in this way will take $\otilde (n^3/\epsilon)$ time. 

However, we can do better by taking the pseudoinverse sketch and leveraging some details of its construction to ``matrix-ize" it. More precisely, we quickly convert our $\otilde (n/\epsilon)$-bit data structure into a $n \times n$ dense matrix $\mathcal{M}$ where querying our data structure with $\x$ is equivalent to computing $\x^\top \mathcal{M} \x$.  As this matrix is an approximate pseudoinverse, we can then use it to read off the effective resistances as we please. It is important to note that this procedure only works in the all-pairs regime. Although one might hope this would be useful for computing leverage scores of all edges, we require $\otilde(n^2/\epsilon)$ work before computing any resistances. In dense graphs however, this is an improvement by an $\epsilon$ factor over the previous work, and it suggests that similar running times may be possible for computing leverage scores across all edges in sparse graphs. 

%% file: lsketch.tex
\section{Sketching Laplacians} 
\label{sec:lap_sketch}

Here we prove Theorem~\ref{thm:main}, i.e. we show how to produce a $\otilde(n/ \epsilon)$-sparse $\epsilon$-spectral sketch of a Laplacian for all $\epsilon > 0$. As explained in Section~\ref{sub:our_approach}, our approach is to partition the graph $G$ into connected components where each connected component lies inside some expander inside $G$ and then sketch each of these connected components separately. In Section \hyperlink{31}{4.1}, we provide the algorithm for each of the individual components and then in Section~\hyperlink{32}{4.2} 
we provide the partitioning procedure and put everything together to prove Theorem~\ref{thm:main}.

\hypertarget{31}{\subsection{Sketching Subsets of Expanders}}
\label{sub:sketch_expander}

Here we construct our basic sketching tool for sketching well-connected subgraphs. The main result of this section is to prove the following result which is a key component of our proof of Theorem~\ref{thm:main}. 

\begin{lemma}[Sketching Subgraphs of Expanders]
\label{lemma:expandersketch}
	Let $H$ be an unweighted undirected graph with $n$ nodes and $m$ edges and let $T$ be some unknown unweighted graph with $n'$ nodes and conductance $h$ which contains $H$ as an induced subgraph. Let $\x \in \mathbb{R}^{n}$ be any vector, and let $\z \in \mathbb{R}^{n'}$ be any other vector such that $\x_u = \z_u$ for all $u \in H$. For any $\alpha \geq 1$, there exists a function $f(\cdot)$ where 
	\[
	\mathbb{E}[f(\z)] = \x^\top \hlap \x
	~
	\text{ and }
	~
	\var{f(\z)} \leq 4\alpha^{-2}h^{-4} (\z^\top  \lap_T \z)^2.
	\]
	Moreover, $f(\cdot)$ can be stored in $\otilde (n \alpha)$ space and it can be constructed in $\otilde(m)$ time. Further, $f$ has $\otilde (n \alpha)$ query time.
\end{lemma}

Our sketching procedure is as follows. Given the graph $H$, we store its degree matrix $\hdiag$ and store all edges incident upon nodes of sufficiently low degree. For the remaining edges (connecting two high-degree nodes), we simply iterate through all of the high degree nodes and sample a small number of edges with replacement from the neighborhoods of each. The pseudocode for this routine is $\texttt{SampleSketch}$ given in Algorithm~\ref{alg:sampler}.

Given the sketch above, we evaluate $f(\x)$ as follows. We begin by taking $\x$ and scaling it by an appropriate multiple of the all-ones vector to give $\y$. Now, since our sketch is effectively a combination of degree matrix and adjacency matrix, we then choose our approximator as $\y^\top \hdiag \y - \y^\top \tilde{\hadj} \y$. The pseudocode for this routine is $\texttt{ExpanderEval}$ given in Algorithms~\ref{alg:exeval}.

Our goal is to show that $\y^\top \tilde{\hadj} \y$ is $\y^\top \hadj \y$ in expectation, the variance of $\y^\top \tilde{\hadj} \y$ can be related to $\y^\top \hdiag \y$, and finally the scaling by the ones vector allows us to show that $(\y^\top \hdiag \y)/(\y^\top \hlap \y)$ is at most a function of the conductance of $T$. These pieces combined give us our claimed expectation and variance bounds and let us prove that $f(\x) =  \texttt{ExpanderEval}(\texttt{SampleSketch}(H,\alpha),\x)$ as given by Algorithms~\ref{alg:sampler} and \ref{alg:exeval} satisfies the claims of Lemma \ref{lemma:expandersketch}. 

Note that this algorithm is effectively the same as the sketching procedure from \cite{quadforms}, and much of the analysis is inherited as well. Our main contribution is a proof of a slightly more general variance bound, which will be important later when enforcing our required conductance guarantees.

In the remainder of this section we prove Lemma~\ref{lemma:expandersketch} in two parts. First in Lemma~\ref{lemma:expander1} we prove our claimed results on $f(\cdot)$'s expectation and variance and then in Lemma~\ref{lemma:expander2} we prove the space and query time claims for the sketch. Lemma~\ref{lemma:expandersketch} follows immediately from these lemmas. 

\begin{algorithm}[t]
	\SetAlgoLined
	\KwIn{Unweighted graph $H = (V, E)$, $|V| \leq n$ and a degree threshold $\alpha$}
	\KwOut{\dat($H$), the data for an $\epsilon$-spectral sketch of $H$.}
	
	$S \gets \{u | u \in V, \delta_u \leq \alpha\}$\;
	$L \gets V \backslash S$\;
	$H_L \gets$ $H$ induced on $L$\;
	\dat($H$)$ \gets \emptyset$\;
	Add $S, L,$ and $\alpha$ to \dat($H$)\;
	Add $\delta_u\ \forall u \in V$ to \dat($H$)\; 
	\For {$u \in \mathcal{S}$}{
		Add all of $u$'s adjacent edges to \dat($H$)\;
	}
	
	\For {$u \in L$}
	{
		Add $\delta_u^L$ to \dat($H$)\;
		$E_u \gets \{ (u, v)\ |\ v \in L \}$\;
		Sample (with replacement) $\alpha$ edges from $E_u$ uniformly\;
		Add the sampled $\alpha$ edges (including repetitions) to \dat($H$)\;}
	\Return{\dat($H$)}
	\caption{$\texttt{SampleSketch}$($H,\alpha$) (Algorithm 6 from \cite{quadforms})} \label{alg:sampler}
\end{algorithm}

\begin{algorithm}[t]
	\SetAlgoLined
	\KwIn{The sketch \dat($H$) returned by Algorithm \ref{alg:sampler}; a vector $\x \in \mathbb{R}^n$}
	\KwOut{An approximation to $\x_H^\top \hlap \x_H$}
	$\hdiag \gets$ the degree matrix found in $\dat(H)$\;
	$c = \frac{\textbf{1}^\top \hdiag \x}{\textbf{1}^\top \hdiag \textbf{1}}$\;
	$\y = \x -  c\textbf{1}$\;
	\Return{
		\[
		\sum_{u \in V(H)} \y_u^2 \delta_u - \sum_{u \in S} \sum_{\substack{v \in V(H) \\ (u,v) \in \dat(H)}} \y_u \y_v -  \sum_{u \in L} \sum_{\substack{v \in S \\ (u,v) \in \dat(H)}} \y_u \y_v - \sum_{u \in L} \frac{\delta^L_u}{\alpha} \sum_{v \in L} \y_u \y_v Y^u_v.
		\]}
	\caption{$\texttt{ExpanderEval}$(\dat($H$),$\x$)}
	\label{alg:exeval}
\end{algorithm}

\begin{lemma}[Expander Sketch - Expectation and Variance]
\label{lemma:expander1}
Let $H$ be an unweighted undirected graph with $n$ nodes and $m$ edges and let $T$ be an unknown unweighted undirected graph with $n'$ nodes and conductance $h$ which contains $H$ as an induced subgraph. Let $\x \in \mathbb{R}^{n}$ be any vector, and let $\z \in \mathbb{R}^{n'}$ be any other vector such that for all $u \in H$ $\x_u = \z_u$. Then, for any $\alpha \geq 1$ our computed $f(\x) = $ \texttt{ExpanderEval}($\texttt{SampleSketch}$($H,\alpha$),$\x$) satisfies $ \mathbb{E}(f(\x)) = \x^\top \hlap \x $ and 
$
\var{f(\x)} \leq 4\alpha^{-2}h^{-4} (\z^\top \lap_T \z)^2
$.
\end{lemma}

\begin{proof} Let $Y^u_v$ be the number of times we sampled the edge $(u,v)$ in line 13 of Algorithm \ref{alg:sampler}. Our algorithm returns the estimator
\[
\sum_{u \in V(H)} \y_u^2 \delta_u - \sum_{u \in S} \sum_{\substack{v \in V(H) \\ (u,v) \in \dat (H)}} \y_u \y_v -  \sum_{u \in L} \sum_{\substack{v \in S \\ (u,v) \in \dat(H)}} \y_u \y_v - \sum_{u \in L} \frac{\delta^L_u}{\alpha} \sum_{v \in L} \y_u \y_v Y^u_v,
\]
where $\y = \x - c\textbf{1}$ for some constant $c$ depending on $\hdiag $ and $\x$ and $\dat(H)$ is the sketch returned by $\texttt{SampleSketch}$. Note that for every node in $S$, all of the edges incident upon it are stored in the sketch. Note that the probability of sampling the edge $(u,v)$ for $u$ is $1/\delta^L_u$. Since we sample $\alpha$ edges for $u$, $Y^u_v$ is distributed as $B(\alpha, 1/\delta^L_u)$, a sum of $\alpha$ trials each with probability of success $1/\delta^L_u$. With this, we see 
\[
\mathbb{E}[Y^u_v] = \frac{\alpha}{\delta^L_u} \text{\quad \quad and \quad \quad} \var{Y^u_v} =  \frac{\alpha}{\delta^L_u}\left(1 - \frac{1}{\delta^{L}_u}\right) \leq \frac{\alpha}{\delta^L_u}.
\]
Further, we see that for any distinct nodes $u, v, w$ $Y_v^u$ and $Y_w^u$ are inversely correlated; choosing $(u,v)$ only decreases the number of times we choose $(u,w)$. Consequently, by linearity of expectation, 
\begin{align*}
\mathbb{E}[f(\x)] &= \sum_{u \in V(H)} \y_u^2 \delta_u - \sum_{u \in S} \sum_{\substack{v \in V(H) \\ (u,v) \in E(H)}} \y_u \y_v -  \sum_{u \in L} \sum_{\substack{v \in S \\ (u,v) \in E(H)}} \y_u \y_v - \sum_{u \in L}  \sum_{\substack{v \in L \\ (u,v) \in E(H)}} \y_u \y_v \\
&= \y^\top \hlap \y = \x^\top \hlap \x,
\end{align*}
since $\hlap$ is orthogonal to the all ones vector.

We will now prove our variance bound. Note that $\hdiag $ and $\x$ are fixed, and further only the last of these four sums in $f$ depends on any random decisions. Therefore,
\begin{align}
\var{f(\x)}  &= \var{\sum_{u\in L}\frac{\delta_u^ L}{\alpha} \sum_{v\in L} \y_u\y_vY_u^v} \label{eq:var}
\leq 
\sum_{u\in L}\frac{(\delta_u^ L)^2}{\alpha^2} \sum_{v\in L} 
\y_u^2\y_v^2\var{Y_u^v} \\
&\leq \sum_{u\in L}\frac{(\delta_u^ L)^2}{\alpha^2} \y_u^2\sum_{v\in L} \y_v^2\frac{\alpha }{\delta_u^ L} =  \frac{1}{\alpha}\sum_{u\in L}\delta_u^ L \y_u^2\sum_{v\in L} \y_v^2\\
&\leq \frac{1}{\alpha}\sum_{u\in L}\delta_u^ L \y_u^2\sum_{v\in L} \y_v^2 \frac{\delta_v}{\alpha} \label{eq:hideg}\\
&\leq \frac{1}{\alpha^2} \sum_{u \in V} \delta_u\y_u^2\sum_{v \in V} \delta_v \y_v^2  =  \frac{1}{\alpha^2} \|D^{1/2}\y\|_2^4,              
\end{align}
where we used in \eqref{eq:var} that for distinct nodes $u$, $v$, and $w$, $Y_v^u$ and $Y_w^u$ are inversely correlated and $Y_w^u$ and $Y_w^v$ are independent, and we used $\alpha < \delta_v$ for all nodes $v \in L$ in \eqref{eq:hideg}. 

So far, our analysis is the same as \cite{quadforms}. However, next, we now bound $\|\hdiag^{1/2}\y\|_2^4$. In the setup of \cite{quadforms}, the conductance of $H$ would be leveraged directly via a Cheeger bound to give a bound in terms of $\y^\top \hlap \y$. However, we know nothing about $H$'s conductance itself-- although we know $T$ containing $H$ has high condutance, $H$ may itself have very poor conductance. We therefore apply Cheeger's inequality in a slightly more roundabout fashion. It can readily be shown that 
\[
c = \textsf{argmin}_k \|\hdiag^{1/2}(\x - k\textbf{1})\|_2.
\]
Therefore, our choice of $c$ when computing $f$ conveniently minimizes our current bound on the variance. With this, define $c' = \frac{1}{\textsf{Vol}(T)} \textbf{1}^\top \mathcal{D}_T \z$ and let $\w = \z -  c' \textbf{1}$. Note $\w$ is orthogonal to $\mathcal{D}_T \textbf{1}$. The minimality of $c$ gives us
\[
\|\hdiag^{1/2}\y\|_2^2  = \|\hdiag^{1/2}(\x - c\textbf{1})\|_2^2 \leq \|\hdiag^{1/2}(\x - c'\textbf{1})\|_2^2 = \sum_{u \in H} \delta_u (\x_{u} - c')^2.
\]
Now, let $\delta_u(T)$ be the degree of $u$ in the unknown expander $T$, and let $\mathcal{D}_T$ be the degree matrix of $T$. For any node $u$ in $H$, $\delta_u \leq \delta_u(T)$ since $H$ is a subgraph of $T$. In addition, we have $\x_{u} = \z_{u}$ by assumption. Therefore, 
\[
\sum_{u \in H} \delta_u (\x_{u} - c')^2 \leq \sum_{u \in H} \delta_u(T) (\z_{u} - c')^2 \leq \sum_{u \in T} \delta_u(T) (\z_{u} - c')^2 = \sum_{u \in T} \delta_u(T) \w_u^2 = \|\mathcal{D}_T^{1/2}\w\|_2^2.
\]
Set $\boldsymbol{\rho} = \mathcal{D}_T^{1/2}\w$. Then, 
\[
\|\mathcal{D}_T^{1/2}\w\|_2^2   = (\boldsymbol{\rho}^\top \boldsymbol{\rho}) \leq \frac{1}{\lambda_2(\mathcal{N}_T)}(\boldsymbol{\rho}^\top  \mathcal{N}_T \boldsymbol{\rho}) = \frac{1}{\lambda_2(\mathcal{N}_T)}(\w^\top  \lap_T \w),
\]
since $\boldsymbol{\rho}$ is orthogonal to $\mathcal{D}_T^{1/2}\textbf{1}$, the zero eigenvector of $N_T$, and now
\[
\frac{1}{\lambda_2(\mathcal{N}_T)}(\w^\top  \mathcal{L}_T \w) \leq \frac{2}{h^2}(\w^\top  \lap_T \w) = \frac{2}{h^2}(\z^\top  \lap_T \z)
\] 
by Cheeger's inequality (Lemma~\ref{lem:cheeger}). Putting this together, we get
\[
\var{f(\x)} \leq \frac{1}{\alpha^2} \|\hdiag^{1/2}\y\|_2^4 \leq \frac{1}{\alpha^2} \|\mathcal{D}_{T}^{1/2}\w\|_2^4 \leq \frac{4}{\alpha^2 h^4}(\z^\top  \lap_T \z)^2 ~.
\]
\end{proof}

\begin{lemma}[Expander Sketch - Space and Query Time]
	\label{lemma:expander2}
	Let $H$ be an unweighted undirected graph with $n$ nodes and $m$ edges. For any $\alpha \geq 1$, $f(\x)$ requires $\otilde(n\alpha)$ bits of space to store, and obtaining $f$ costs $\otilde(m)$ time. Querying $f$ can be performed in $\otilde(n\alpha)$ time. 
\end{lemma}

\begin{proof} We first show the space bound. Note that the storage cost of $f$ is just the amount of space needed to store $\texttt{SampleSketch}$($H,\alpha$). We consider the amount of space each object stored in the sketch costs. $\texttt{SampleSketch}$ stores $n$ degrees, $\alpha, S$ and $L$; this clearly cost $O(n)$ in space. In addition, we store $O(n)$ different $\delta_u^L$, which costs $O(n)$ as well. Now, for every node  $u$ of degree less than $\alpha$, we store all $\delta_u \leq \alpha$ edges incident to it; this costs $\otilde(n\alpha)$. Finally, for every node $v$ of degree more than $\alpha$, we store $\alpha$ edges samples from its neighborhood. This too costs $\otilde(n\alpha)$., proving our space bound.

We now look at the time needed to find $f$. Note that again this is just the amount of time $\texttt{SampleSketch}$ needs to compute its samples. Clearly, computing the $\delta_u$, the $\delta_u^L$, $\alpha$, $S$, and $L$ can all be done in $\otilde(m)$ time. In addition, storing the edges incident upon every node of degree less than $\alpha$ can be done in $\otilde(m)$ time as well. Finally, sampling $\alpha$ edges from a high-degree node $v$'s neighborhood can be done in $\otilde(\alpha)$ time. Since there are at most $O(m/\alpha)$ nodes of degree more than $\alpha$, this entire sketch can be computed in $\otilde(m)$ time, as desired. 

Lastly, we look at the cost of querying $f$. Computing $f(\x)$ requires first computing $y =  x - c\textbf{1}$, where $c = \frac{\textbf{1}^\top \hdiag x}{\textbf{1}^\top \hdiag \textbf{1}}$. This can obviously be done in $O(n)$ time. Then, we compute 
\[ \sum_{u \in V(H)} \y_u^2 \delta_u - \sum_{u \in S} \sum_{\substack{v \in V(H) \\ (u,v) \in \dat(H)}} \y_u \y_v -  \sum_{u \in L} \sum_{\substack{v \in S \\ (u,v) \in \dat(H)}} \y_u \y_v - \sum_{u \in L} \frac{\delta^L_u}{\alpha} \sum_{v \in L} \y_u \y_v Y^u_v.
\] 
By the above space analysis, this expression sums $\otilde(n\alpha)$ terms, and therefore can be computed in $\otilde(n\alpha)$ time. Note that if $n\alpha > m$ then we can simply use the entire graph as the sketch and the result follows trivially.
\end{proof}

\hypertarget{32}{\subsection{Sketching General Graphs}}
\label{sub:sketch_graph}

With Lemma~\ref{lemma:expandersketch} established here we prove Therorem~\ref{thm:main} by partitioning an arbitrary graph $G$ into pieces satisfying Lemma~\ref{lemma:expandersketch}'s condition and sketching each piece separately. We make extensive use of a graph partitioning result of \cite{ST} on partitioning a graph into expander-like piece. In particular we use the following variant: 

\begin{lemma}[Lemma 32, \cite{KLOS}]
	Given an unweighted graph $G$, we can construct in time $\tilde O(m)$ a partition of the nodes $V_1, V_2, ... V_k$ and a sequence of subsets $S_1, S_2, ... S_k$ where
	\begin{itemize}
		\item $S_i \subseteq V_i,$ $\forall i$.
		\item The $S_i$ contain at least half the edges: $\sum |E[S_i]| > \frac{1}{2}m$.
		\item For every $i$, there exists a set $T_i$ where $S_i \subseteq T_i \subseteq V_i$ and $G(T_i)$ has conductance $\Omega(1/\log^2 n)$.
	\end{itemize}
	\label{lemma:part}
\end{lemma}

This partitioning scheme above creates sets $S_i$ where each $S_i$ is inside a good expander $T_i$, the $T_i$ are disjoint, and the $S_i$ contain half the edges in the graph. With this, we can take a graph $H$, apply the theorem to get $S_i$ lying inside expanders in $H$, and then recurse on the graph of edges which cross between different $S_i$. Since half the volume lies inside the $S_i$ at every step, this scheme generates $O(\log n)$ levels $H_1, H_2, ... H_k$ which partition the edges of $H$. Further, in a given level $H_i$ every connected component lies inside some expander of conductance $\Omega(1/\log^2 n)$ in $\cup_{j \leq k} H_j$. 

However, the above partitioning scheme only works on unweighted graphs. We can get around this in the polynomially-bounded integer weights case with standard tricks as in \cite{ST}. Assume the maximum weight of an edge in our graph $G$ is $M$. Define the unweighted graphs $G_i$, $i = 1, 2, ... \lfloor{\log M}\rfloor + 1$ as follows: $G_i$ is on the same vertices as $G$, and $u$ and $v$ are connected in $G$ if the $i^{th}$ bit of the binary expansion of the weight of $(u,v) \in E(G)$ is 1. Although each $G^i$ is unweighted, if we treat every edge in every $G_i$ as having weight $1$ we have
\[
\glap = \sum_{i=1}^{\lfloor{\log M}\rfloor + 1} 2^{i-1}\mathcal{L}_{G_i}.
\]
In addition, we see that all the $G_i$s have $n$ nodes and at most $m$ edges, and since $M$ is polynomially bounded there are at most $O(\log n)$ of them. Sparsifying each of these separately and adding up the estimators we get will then give us an estimator for our original weighted graph $G$.

In light of these considerations, consider Algorithm \ref{alg:partition}. This partitions every $G_i$ into $O(\log n)$ subgraphs $G_{i1}, G_{i2}, ... G_{i O(\log n)}$, where the $G_{ij}$ edge partition $G_i$ and each of $G_{ij}$'s connected components lies in a unique expander in $\cup_{j \leq l} G_{il}$. Note that there are $O(\log^2 n)$ $G_{ij}$. It then separates each of the $G_{ij}$ into connected components, and returns the set of connected components. Now, if $H_{ijk}$ is a connected component of $G_{ij}$ and if $\x_{H_{ijk}}$ is our query vector $\x$ induced on the nodes of $H_{ijk}$, we have
\[
\x^\top \glap \x = \sum_{i=1}^{\lfloor{\log M}\rfloor + 1} 2^{i-1}\sum_{j} \sum_{k = 1}^{\# \text{ connected components of } G_{ij }} \x_{H_{ijk}}^\top  \lap_{H_{ijk}} \x_{H_{ijk}}.
\]
But now, each connected component $H_{ijk}$ is known to lie inside some expander subgraph of $\cup_{j \leq k} G_{ij}$. With this, we can apply Lemma \ref{lemma:expandersketch} and sketch every one of these connected components, and then summing the contribution of all of the sketches yields something that can approximate $\x^\top \glap \x$.

\begin{algorithm}[t]
	\caption{$\texttt{Split}$($G$)} \label{alg:partition}
	\KwIn{A polynomially-bounded weighted graph $G = (V,E)$}
	\KwOut{An expander partitioning of $G$}
	\textsf{part}($G$) $\gets \emptyset$\;
	Bitbucket the edges of $G$ to form unweighted graphs $G_1, G_2...$ satisfying
	\[
	\glap = \sum_i 2^{i-1} \lap_{G_i}
	\]
	\For{\textsf{each} $G_i$}{
		$j = 1$\;
		\While{$G_i \neq \emptyset$}{
			Apply Lemma \ref{lemma:part} to $G_i$ to get $S_1, S_2, ... $ \;
			$E_{ij} \gets \{(u,v)\ | (u,v) \in G_i, \exists k: u \in S_k, v \in S_k \}$ \;
			$G_{ij} \gets (V_i, E_{ij})$ \;
			$G_{i} \gets G_i \backslash G_{ij}$ \;
			$j \gets j+1$ \;
		}
	}
	\For{\textsf{each} $G_{ij}$}{
		$H_{ij1}, H_{ij2}, ... \gets$ the connected components of $G_{ij}$
	}

	\For{\textsf{each} $H_{ijk}$}{
		Add $(i, H_{ijk})$ to \textsf{part}($G$)
	}
	
	\Return{\textsf{part}($G$)}
\end{algorithm}
\begin{algorithm}[ht]
	\SetAlgoLined
	\KwIn{A polynomially-bounded weighted graph $G$, an error tolerance $\epsilon$}
	\KwOut{\textsf{data}($G$) Data for $\epsilon$-spectral sketch}
	$\mathbb{G} \gets \texttt{Split}(G)$\;
	$\alpha \gets O(\log^{4.5} n / \epsilon)$\;
	$\textsf{data}(G) \gets \emptyset$\;
	\For{\upshape each $(i, H_{ijk}) \in \mathbb{G}$}{
		Add ($i$, $\texttt{SampleSketch}$($H$, $\alpha$)) to $\textsf{data}(G)$
	}
	\Return{\textsf{data}($G$)}
	
	\caption{$\texttt{Sketch}$($G$, $\epsilon$)}
	\label{alg:sketch}
\end{algorithm}

\begin{algorithm}
	\SetAlgoLined
	\KwIn{$\textsf{data}(G)$ a sketch of a graph $G$, $\x$ a query vector}
	\KwOut{an approximation to $\x^\top \glap \x$}
	$r \gets 0$\;
	\For{\upshape each integer-expander sketch pair $(i,\dat(H))$ in \textsf{data}($G$)}{
		$\x_H \gets \x$ induced on the vertices of $H$\;
		$r \gets r + 2^{i-1}\ \texttt{ExpanderEval}(\dat(H),\x_H)$\;
	}
	\Return{r}
	\caption{$\texttt{Eval}$($\textsf{data}(G), \x$)}
	\label{alg:eval}
\end{algorithm}

Consider Algorithms \ref{alg:sketch} and \ref{alg:eval} and define $g(\x) = $ $\texttt{Eval}$($\texttt{Sketch}$($G, \epsilon$),$\x$). We prove the following lemma about $g(\x)$:

\begin{lemma} \label{lemma:epsketch}
	Let $G$ be a weighted undirected graph with $n$ nodes and $m$ edges. Let $\epsilon > 0$ be a parameter, and let $\x$ be a query vector. $g(\x)$ is an unbiased estimator for $\x^\top \glap \x$, and 
	\[
	\frac{\var{g(\x)}}{\epsilon^2 \mathbb{E}[g(\x)]^2} \leq 0.1
	\]
	Further, $g(\x)$ can be stored in $\otilde(n/\epsilon)$ space, and it can be found in $\otilde(m)$ time. A query to $g(\x)$ can be answered in $\otilde(n/\epsilon)$ time. 
\end{lemma}

Before we prove this result, we first explain how Lemma \ref{lemma:epsketch} would imply Theorem \ref{thm:main}.

\begin{proof}[Proof of Theorem \ref{thm:main}]
Assume that Lemma \ref{lemma:epsketch} held. Note that $g(\x)$ satisfies the desired bounds on constructing time, storage space, and query time. We now show $g(\x)$ satisfies the approximation guarantee $(1-\epsilon) \x^\top \glap \x \leq g(\x) \leq (1+\epsilon) \x^\top \glap \x$ with probability $0.9$. By Chebyshev's inequality we have
\[
\mathbb{P}\left(|g(\x) - \x^\top \glap \x| > \epsilon \cdot \x^\top \glap \x\right) 
= \mathbb{P}\left(|g(\x) - \mathbb{E}[f(\x)]| 
> \epsilon \cdot \mathbb{E}[g(\x)]\right) 
< \frac{\var{g(\x)}}{\epsilon^2 \mathbb{E}[g(\x)]^2}.
\]
Since we know the term on the right is at most $0.1$ by the assumed lemma, our function $g(\x)$ satisfies $(1-\epsilon) \x^\top \glap \x \leq g(\x) \leq (1+\epsilon) \x^\top \glap \x$ with probability $0.9$, as claimed. 
\end{proof}

Consequently, all that remains to complete this section is to prove Lemma \ref{lemma:epsketch}.

\begin{proof}[Proof of Lemma \ref{lemma:epsketch}] We will prove the claims in order. For notational clarity, define $\lap^{ijk} = \lap_{H_{ijk}}$ and $\x_{ijk} = \x_{H_{ijk}}$. We start with unbiasedness. Let $\alpha = c \log^{4.5} n / \epsilon$, for some constant $c$ defined later. Note that in the end 
\[
g(\x) = \sum_{i=1}^{\lfloor{\log M}\rfloor + 1} 2^{i-1}\sum_{j} \sum_{k = 1}^{\# \text{ connected components of } G_{ij }}  f_{ijk}(\x^{ijk}).
\]
where $f_{ijk}(\x_{ijk}) =$ $\texttt{ExpanderEval}$($\texttt{SampleSketch}$($H_{ijk}, \alpha$), $\x_{ijk})$ is just the algorithm from Lemma \ref{lemma:expandersketch}. Since $f_{ijk}$ is unbiased, $\mathbb{E}[f_{ijk}(\x_{ijk})] = \x_{ijk}^\top \lap^{ijk}\x_{ijk}$, and
\[
\mathbb{E}[g(\x)] = \sum_{i=1}^{\lfloor{\log M}\rfloor + 1} 2^{i-1}\sum_{j} \sum_{k}  \x_{ijk}^\top \lap^{ijk}\x_{ijk} = \x^\top \glap \x
\]
as desired. We now turn our attention to the inequality. Note by the definition of the partitioning, each $H_{ijk}$ lies inside a subgraph $T_{ijk}$ of  $\cup_{j \leq l} G_{il}$ with conductance $\Omega(\frac{1}{\log^2 n})$, and further none of $T_{ij1}, T_{ij2}, ... $ overlap. Define $\z_{ijk}$ as $\x$ induced on the nodes of $T_{ijk}$. As $\x_{ijk}$ and $\z_{ijk}$ obviously overlap on the nodes in $H_{ijk}$, by Lemma \ref{lemma:expandersketch}, we can then say
\[
\var{f_{ijk}(\x_{ijk})} \leq O(\alpha^{-2}\log^{8} n )(\z_{ijk}^\top \lap_{T_{ijk}}\z_{ijk})^2.
\]
As the $f_{ijk}$ are all independently generated, bringing the variance inside the sum and applying our bound gives
\[
\var{g(\x)} \leq \sum_{i=1}^{\lfloor{\log M}\rfloor + 1} 4^{i-1}\sum_{j} \sum_{k} O(\alpha^{-2}\log^{8} n )(\z_{ijk}^\top \lap_{T_{ijk}} \z_{ijk})^2.
\]
Now, note that the $T_{ijk}$ are subgraphs of $\cup_{j \leq l} G_{il} \subseteq G_i$, and further for fixed $i$ and $j$ the $T_{ijk}$ do not overlap. As such, the vectors $\z_{ijk}$ appear as distinct subvectors of $\x$, and we can write
\[
\sum_{k} (\z_{ijk}^\top \lap_{T_{ijk}} \z_{ijk})^2 \leq (\sum_{k} \z_{ijk}^\top \lap_{T_{ijk}} \z_{ijk})^2 \leq (\x^\top \lap_{G_i} \x)^2.
\]
Plugging this in gives 
\[
\var{g(\x)} \leq \sum_{i=1}^{\lfloor{\log M}\rfloor + 1} 4^{i-1}\sum_{j}  O(\alpha^{-2}\log^{8} n )(\x^\top \lap_{G_i}\x)^2.
\]
As there are $O(\log n)$ distinct values of $j$ in the inner sum, we also have
\[
\var{g(\x)} \leq \sum_{i=1}^{\lfloor{\log M}\rfloor + 1} 4^{i-1} O(\alpha^{-2} \log^9 n)(\y^\top  \lap_{G_i} \y)^2 = \sum_{i=1}^{\lfloor{\log M}\rfloor + 1} 4^{i-1} O(c^{-2} \epsilon^{2})(\y^\top \lap_{G_i} \y)^2.
\]
Obeserve 
\[
\mathbb{E}[g(\x)] =  \x^\top  \glap \x = \sum_{i=1}^{\lfloor{\log M}\rfloor + 1} 2^{i-1} \x^\top \lap_{G_i} \x.
\]
With this, we can say
\[
\mathbb{E}[g(\x)]^2 \geq \sum_{i=1}^{\lfloor{\log M}\rfloor + 1} 4^{i-1} (\x^\top \lap_{G_i} \x)^2
\]
and so
\[
\frac{\var{g(\x)}}{\epsilon^2 \mathbb{E}[g(\x)]^2} \leq \frac{\sum_{i=1}^{\lfloor{\log M}\rfloor + 1} 4^{i-1} O(c^{-2})(\x^\top  \lap_{G_i} \x)^2}{\sum_{i=1}^{\lfloor{\log M}\rfloor + 1} 4^{i-1}(\x^\top \lap_{G_i} \x)^2} = \frac{O(1)}{c^2}.
\]
Choosing $c$ appropriately implies the result.

We now look at the amount of space needed to store $g(\x)$. Since we store one Lemma \ref{lemma:expandersketch} sketch per $H_{ijk}$, we will use $\otilde(|H_{ijk}|/ \epsilon)$ space per $H_{ijk}$. Now, for fixed $i$ and $j$, $\sum_k |H_{ijk}| \leq n$. Therefore, storing all of the sketches for fixed $i$ and $j$ costs $\otilde(n/\epsilon)$ space. Finally, as there are $O(\log^2 n)$ values for $(i,j)$, the space requirement follows.

We next look at the running time of this procedure. Note that we can recursively compute the partitioning from Algorithm \ref{alg:partition} in $\otilde(m)$ time. Further, by Lemma \ref{lemma:expandersketch} we will take $\otilde(|E(H_{ijk})|$ time to sketch each $H_{ijk}$. Since $\sum_k |E(H_{ijk})| \leq m$, and since there are at most $O(\log^2 n)$ values for $(i,j)$, this claim follows as well.

Finally, we look at the query time of the functon $f$ we computed. Similar to the space bound proven above, we see that by Lemma \ref{lemma:expandersketch} inducing over the proper nodes and querying the stored sketches for fixed $i$ and $j$ takes $\otilde(n/\epsilon)$ time. Since there are $O(\log^2 n)$ different values for $(i,j)$, querying the entire sketch takes $\otilde(n/\epsilon)$ time as well.
\end{proof}


%% file: pseudosketch.tex
\section{Sketching the Pseudoinverse}
\label{sec:pseudo}

In Section~\ref{sec:lap_sketch}, we provided an algorithm which computes an $\epsilon$-spectral sketch for a Laplacian using $\otilde(n/\epsilon)$ space. 
%
In this section we show that this machinery can also be used to achieve a similar result for sketching the pseudoinverse and prove Theorem~\ref{thm:pinv}, restated below. 

\pinv*

We prove this result by providing a much more general reduction.  Informally we show that to compute an $\epsilon$-spectral sketch for the pseudoinverse it suffices both a standard $\sqrt{\epsilon}$-spectral approximation and a $\epsilon$-spectral sketch. Since $\sqrt{\epsilon}$-sparsifiers for the Laplacian and its pseudoinverse can be stored in $\otilde(n/\epsilon)$ space, this yields our desired result. Formally, we show the following.

\begin{theorem}[Pseudoinverse Sketch Reduction] \label{thm:general}
Let $A \in \R^{n \times n}$ be a symmetric positive semidefinite (PSD) matrix, let $\bv \in \Range(A)$, let  $S \in \R^{n \times n}$ satisfy $(1 + \sqrt{\epsilon})^{-1} A^\pseudo \preceq S \preceq (1 + \sqrt{\epsilon}) A^\pseudo$ for $\epsilon \leq 1/16$, and let $\y = S \bv$. If $f(\y)$ is some value satisfying  $(1 - \epsilon) \y^\top A \y \leq f(\y) \leq (1 + \epsilon) \y^\top A \y$, then 
\[
|2\bv^\top \y - f(\y) -  \bv^\top A^\pseudo \bv| \leq 4\epsilon \cdot \bv^\top A^\pseudo \bv
\]
\end{theorem}

Note that $(1 + \sqrt{\epsilon})^{-1} A^\pseudo \preceq S \preceq (1 + \sqrt{\epsilon}) A^\pseudo$ if and only if $(1 + \sqrt{\epsilon}) A \succeq S^\pseudo \succeq (1 + \sqrt{\epsilon})^{-1} A$, i.e. $S$ is a $\sqrt{\epsilon}$-spectral approximation to $A^\pseudo$ if and only if $S^\pseudo$ is a $\sqrt{\epsilon}$-spectral approximation to $A$. Consequently, this theorem implies that a $\sqrt{\epsilon}$-sparsifier for $A$ and an $\epsilon$-sketch for $A$ together constitute a $4\epsilon$-spectral sketch for $A^\pseudo$. Moreover, the cost of evaluating this sketch of $A^\pseudo$ is essentially the cost of evaluating the sketch of $A$ and solving a linear system in the sparsifier of $A$. 

The intuition behind our approach is as follows. Given a positive semidefinite matrix $A$ and $\bv \in \Range(A)$, consider the function $q_{A,\bv}(\z) = 2\bv^\top \z - \z^\top A \z$. Maximizing this objective function in $\z$ is a standard problem in numerical linear algebra; it is the objective function steepest descent and conjugate gradient seek to maximize to solve the linear system $A \x = \bv$ \cite{golub}. It is not hard to show (see Lemma~\ref{lemma:qfacts}) that the maximizer of this function is  $\x = A^\pseudo \bv$, and it achieves objective value $\bv^\top A^\pseudo \bv$, exactly the quantity we want to approximate. Consequently, to obtain a sketch of $A^\pseudo$ we simply need to maintain the ability to approximately maximize $q_{A,\bv}(\z)$ for any input $\bv$. 

To prove Theorem~\ref{thm:general} we show how to approximate $q_{A,\bv}(\x)$, i.e. approximately maximize $q_{A,\bv}$, through two approximations. First, rather than actually maximizing  $q_{A,\bv}$, i.e. solving $A \x = \bv$ and querying $q_{A,\bv}(\x)$, we instead find $\y = S \bv$, where $S$ is $\sqrt{\epsilon}$-spectrally close to $A^\pseudo$, and query $q_{A,\bv}(\y)$. In other words, we solve $A' \y = \bv$ where $A' = S^\pseudo$ is a $\sqrt{\epsilon}$ approximation to $A$.  The second approximation, is that we replace the quadratic form $ \y^\top A \y$ inside $q_{A,\bv}(\y)$ with $f(\y)$ which is $\epsilon$-close to $\y^\top A \y$. We show that the error incurred by each of these substitutions is small simply by carefully analyzing $q_{A,\bv}$ (see Lemma~\ref{lemma:qfacts} and Lemma~\ref{lemma:approxfact}) and combining these bounds to prove Theorem~\ref{thm:general}.

To actually prove our theorem for Laplacians, i.e. use Theorem~\ref{thm:general} to prove Theorem~\ref{thm:pinv}, we show that the required inputs to Theorem~\ref{thm:general} can be computed, stored, and applied in the required time and space. This follows directly from the extensive literature on Laplacian system solving and graph sparsification (though it is implicit in the work and some care is needed to obtain the desired form). 

We begin with some useful properties of the function $q_{A,\bv}(\x)$. Some of these properties are well known, but we include their proof for completeness.

\begin{lemma}[Basic Facts about $q_{A, \bv}$] \label{lemma:qfacts}
Let $A \in \R^{n \times n}$ be a symmetric (PSD) matrix, let $\bv \in \Range(A)$, and define $q_{A,\bv}(\y) = 2\bv^\top \y - \y^\top A \y$ for all $\y$. Then if $\x \in \R^n$ satisfies $A \x = \bv$ we have 
\[
q_{A,\bv}(\x) - q_{A,\bv}(\y) = (\x - \y)^\top  A (\x - \y).
\] 
for all $\y \in \R^{n}$. Further, $\x$ maximizes $q_{A, \bv}$ and $q_{A, \bv}(\x) = \bv^\top A^\pseudo \bv$. 
\end{lemma}
\begin{proof}
First, we see by direct calculation
\[
q_{A, \bv}(\x) = 2\bv^\top \x - \x^\top A \x = 2\x^\top A \x - \x^\top A \x =\x^\top A \x = \bv^\top A^\pseudo A A^\pseudo \bv = \bv^\top A^\pseudo \bv ~.
\]
yielding the final claim. Next, we see by direct calculation
\[
q_{A,\bv}(\x) - q_{A,\bv}(\y) 
=  \y^\top A \y - \x^\top A \x  + 2 \bv^\top(\x - \y) ~.
\]
However, since $A \x = \bv$ this implies
\[
q_{A,\bv}(\x) - q_{A,\bv}(\y) =  \y^\top A \y + \x^\top A \x - 2 \x^\top A \y = (\x - \y)^\top  A (\x - \y)
\]
yielding the first result. To see that $\x$ maximizes $q_{A, \bv}$, note that for any $\y$ $(\x - \y)^\top  A (\x - \y)$ is nonnegative by since $A$ is PSD. Therefore for any vector $\y$ we have $q_{A,\bv}(\x) \geq q_{A,\bv}(\y)$. 
\end{proof}

In the above lemma, the difference $q_{A,\bv}(\x) - q_{A,\bv}(\y)$ is a quadratic form in $A$. We show that this term is well behaved for $\y$ derived from matrices spectrally close to $A$. Again, this is analogous to previous work on sampling algorithms for regression.

\begin{lemma}[Error from Solving an Approximate System]
\label{lemma:approxfact}
Let $A, A' \in \R^{n \times n}$ be symmetric positive semidefinite matrices, let $\bv \in \Range(A)$, and define $q_{A,\bv}(\z) = 2\bv^\top \z - \z^\top A \z$ for all $\z \in \R^n$. If for $\x,\y \in \R^n$ we have $A \x = A' \y = \bv$ and $(1 + \sqrt{\epsilon})^{-1} A \preceq A' \preceq (1 + \sqrt{\epsilon}) A$ for $\epsilon \leq 1/16$, we have
\[
\y^\top A \y \leq 2 q_{A, \bv}(\x)
~ \text{ and } ~
(\x - \y)^\top A (\x - \y) \leq 2 \epsilon q_{A, \bv}(\x)
~. 
\]
\end{lemma}
\begin{proof}

Note that as $(1 - \sqrt{\epsilon}) A \preceq (1 + \sqrt{\epsilon})^{-1} A \preceq A'$, we have 
\[
q_{A,\bv}(\y) = 2 \bv^\top \y - \y^\top A \y = 2 \y^\top A' \y - \y^\top A \y
\geq \left[2 (1 - \sqrt{\epsilon}) - 1\right] \y^\top A \y
= (1 - 2 \sqrt{\epsilon})  \y^\top A \y ~.
\]
However, since $\x$ is the maximizer of $q_{A,\bv}$ by Lemma \ref{lemma:qfacts}, we have $q_{A,\bv}(\y) \leq q_{A,\bv}(\x)$, and since $1 - 2\sqrt{\epsilon} \geq 1/2$, we have $\y^\top A \y \leq 2 q_{A, \bv}(\x)$ as desired. 

To prove the second claim, let $A^{1/2}$ denote the square root of $A$, the unique positive semidefinite matrix such that $A^{1/2} A^{1/2} = A$, and let $A^{-1/2}$ denote the pseudoinverse of $A^{1/2}$. We have
\begin{align*}
(\x - \y)^\top A (\x - \y) 
&=
\|A^{1/2} (\x - \y)\|_2^2
= \|A^{-1/2} A (\x - \y)\|_2^2
= \|A^{-1/2} (A - A') \y\|_2^2\\
&= \|A^{-1/2} (A - A') A^{-1/2} A^{1/2} \y \|_2^2
\leq \|A^{-1/2} (A - A') A^{-1/2}\|_2^2  \cdot \|A^{1/2} \y \|_2^2
\end{align*}
where in the last line we used that the formula does not change if we add to $y$ anything in the kernel of $A$. Since $(1 + \sqrt{\epsilon})^{-1} A \preceq A' \preceq (1 + \sqrt{\epsilon}) A$, by standard properties of $\preceq$ (see \cite{golub}) we have
\[
- \left( \sqrt{\epsilon}\right) I \preceq - \left(\frac{\sqrt{\epsilon}}{1+\sqrt{\epsilon}}\right) I \preceq A^{-1/2} (A' - A) A^{-1/2} \preceq \left(\sqrt{\epsilon}\right) I
\]
and therefore $ \|A^{-1/2} (A - A') A^{-1/2}\|_2^2 \leq \epsilon $. By applying our above bound on  $\|A^{1/2} \y \|_2^2 = \y^\top A \y$, the result follows. 
\end{proof}
With Lemma~\ref{lemma:qfacts} and Lemma~\ref{lemma:approxfact} in hand, we prove Theorem~\ref{thm:general}:
\begin{proof}[Proof of Theorem~\ref{thm:general}]
Let $A' = S^\pseudo$.  Note that this implies $(1 + \sqrt{\epsilon})^{-1} A \preceq A' \preceq (1 + \sqrt{\epsilon}) A$. Also, note that $\Range(A') = \Range(S) = \Range(A^\pseudo) = \Range(A)$ (the middle equality holds because $S$ is a two-sided approximation for $A^\pseudo$) and therefore $A' \y = \bv$. Define $\x = A^\pseudo \bv$, and note that this gives $A \x = A'\y = \bv$. Consequently, letting $q_{A,\bv}(\z) = 2\bv^\top \z - \z^\top A \z$ for all $\z \in \R^n$ we see that the preconditions for Lemmas \ref{lemma:qfacts} and \ref{lemma:approxfact} hold, yielding
\[
0 \leq q_{A,\bv}(\x) - q_{A,\bv}(\y) 
=  (\x - \y)^\top  A (\x - \y) \leq 2 \epsilon \cdot q_{A, \bv}(\x) ~.
\] 
However, since $q_{A,\bv}(x) = \bv^\top A \bv$ by Lemma~\ref{lemma:qfacts} taking absolute values implies 
\begin{equation}
|\bv^\top A \bv - 2 \bv^\top \y + \y^\top A \y| =
|q_{A,\bv}(\x) - q_{A,\bv}(\y)| 
\leq 2 \epsilon \cdot q_{A, \bv}(\x)
\leq 2 \epsilon \cdot \bv^\top A \bv \label{eq:b1} ~.
\end{equation}
Now, by our assumption on $f$ 
we have $|f(\y) - \y^\top A \y| \leq \epsilon \y^\top A \y$ and by Lemma \ref{lemma:approxfact} with $A' = S^\pseudo$, we have $\y^\top A \y \leq 2 q_{A, \bv} (\x) = 2 \bv^\top A \bv$. Combining yields 
\begin{equation}
|\y^\top A \y - f(\y)| \leq 2 \epsilon \cdot \bv^\top A \bv \label{eq:b2}
\end{equation}
Combining \eqref{eq:b1} and \eqref{eq:b2} we get
\begin{align*}
|2\bv^\top \y - f(\y) -  \bv^\top A^\pseudo \bv|
&\leq |2\bv^\top \y - f(\y) - q_{A, \bv} (\y) | + |q_{A, \bv} (\y) - q_{A, \bv} (\x)| \\ 
&\leq  2 \epsilon \cdot \bv^\top A \bv + 2\epsilon \cdot \bv^\top A \bv = 4\epsilon \cdot \bv^\top A \bv. &\qedhere
\end{align*}
\end{proof}

With Theorem \ref{thm:general} in hand to prove Theorem~\ref{thm:pinv} we simply need to show that the sketch $f(\cdot)$ and the matrix $S$ can be computed and applied quickly and stored in low space. By Theorem~\ref{thm:main} proven in Section~\ref{sec:lap_sketch} we already know that we can achieve an $f$ with the requisite properties. In the following lemma we show that we can compute a satisfactory $S$ as well. This lemma is implicit in previous work on preconditioning, linear system solving, and Laplacian system solving, however we include a more detailed proof for completeness.


\begin{lemma}[Construction of Approximate Laplacian Solver Operator] \label{lemma:solver}
There is an algorithm, which when given a Laplacian $\glap \in \R^{n \times n}$ for a graph with  $m$ edges and an error tolerance $\epsilon \in (0,1)$, computes with probability at least $0.9$ in time $\otilde(m)$ an implicit representation of a symmetric matrix $S \in \R^{n \times n}$ such that
\begin{itemize}
\item $(1 + \sqrt{\epsilon})^{-1} \glap^\pseudo \preceq S \preceq (1 + \sqrt{\epsilon}) \glap^\pseudo$, 
\item the representation can be stored in $\otilde(n/\epsilon)$ space, and
\item for any vector $\x$, we can compute $S \x$ in $\otilde(n/\epsilon)$ time.
\end{itemize} 
\end{lemma}

Again, note that there are many ways to prove this lemma and it follows from a general understanding of the literature on Laplacian system solving. Here we provide a particular general approach towards the lemma that allows us to use sparsification and Laplacian system solvers fairly generally.

\begin{proof}[Proof of Lemma~\ref{lemma:solver}]
We construct $S$ as follows. First, we compute a Laplacian matrix $\glap'$ which is a $\sqrt{\epsilon}/4$-spectral sparsifier of $\glap$. With $\glap'$, we compute a symmetric linear operator $N$ satisfying $\frac{1}{2} N^\pseudo \preceq \glap' \preceq 2 N^\pseudo$. Finally, for $z = \lceil 4 \log(16/\sqrt{\epsilon}) \rceil$ we let
\[
S = \frac{1}{2} N \sum_{k = 0}^{z} \Big(I - \frac{\glap' N}{2}\Big)^k .
\]
We claim that this $S$ has all of the desired properties. 

We begin by showing that $S$ spectrally approximates $\glap^\pseudo$. This follows from general analysis of preconditioning and our particular proof is taken from components of \cite{LeeS15}. First, we see that since $\glap'$ spectrally approximates $\glap$, if we can instead show that $S$ sufficiently closely approximates $\glap'^\pseudo$, we can show that it also approximates $\glap$.  In this light, we will first show $(1 - \sqrt{\epsilon}/4) \glap'^\pseudo  \preceq S \preceq \glap'^\pseudo.$

We first note that we can write  $S$ equivalently as 
\[
S = \frac{1}{2} N^{1/2} \sum_{k = 0}^{z} \Big(I - \frac{1}{2} N^{1/2} \glap' N^{1/2} \Big)^k N^{1/2} ~.
\]
Let $\x$ be any vector, and let $\y = N^{1/2} \x$. By the definition of $N$, we have $\y \in \Range(N^{1/2}) = \Range(N) = \Range(\glap')$. Further note that for any vector $\w \in \Range(\glap')$, by $N$'s spectral guarantee we have $\frac{1}{4} \w^\top \w \preceq  \frac{1}{2} \w^\top N^{1/2} \glap' N^{1/2} \w   \preceq  \w^\top \w$. Now,
\begin{gather*}
\frac{1}{2} \x^\top N^{1/2} \sum_{k = 0}^{\infty} \Big(I - \frac{N^{1/2} \glap' N^{1/2}}{2}\Big)^k N^{1/2} \x = \frac{1}{2} \y^\top \sum_{k = 0}^{\infty} \Big(I - \frac{N^{1/2} \glap' N^{1/2}}{2}\Big)^k \y \\
= \frac{1}{2} \x^\top N^{1/2} \Big(\frac{N^{1/2} \glap' N^{1/2}}{2}\Big)^\pseudo N^{1/2} \x  = \x^\top\glap'^\pseudo \x.
\end{gather*}
Therefore, 
\[
\frac{1}{2} N^{1/2} \sum_{k = 0}^{\infty} \Big(I - \frac{N^{1/2} \glap' N^{1/2}}{2}\Big)^k N^{1/2} = \glap'^\pseudo.
\]
With this, we have
\[
\glap'^\pseudo - S = \frac{1}{2} N^{1/2} \sum_{k = z+1}^{\infty} \Big(I - \frac{N^{1/2} \glap' N^{1/2}}{2}\Big)^k N^{1/2} \succeq 0.
\]
We now prove the lower bound on $S$. By the same idea as before, let $\x$ be any vector, and let $\y = N^{1/2} \x$. Note that $\y \in \Range(\glap')$, and further for any vector $\w \in  \Range(\glap')$ \[
\w^\top \Big(I - \frac{N^{1/2} \glap' N^{1/2}}{2}\Big) \w \leq \frac{3}{4} \w^\top \w
\] by the definition of $N$. With this, it is not hard to see that for any integer $k$ we have 
\[\w^\top \Big(I - \frac{N^{1/2} \glap' N^{1/2}}{2}\Big)^k \w \leq \left(\frac{3}{4}\right)^k \w^\top \w
\]
 as well and so
\[
\x^\top (\glap'^\pseudo - S) \x = \frac{1}{2} \y^\top \Big( \sum_{k = z+1}^{\infty} \Big(I - \frac{N^{1/2} \glap' N^{1/2}}{2}\Big)^k \Big) \y 
\leq \frac{1}{2} \sum_{k = z+1}^{\infty} \Big(\frac{3}{4}\Big)^k (\y^\top \y) \leq \frac{\sqrt{\epsilon}}{8} (\x^\top N \x)
\]
by our choice of $z$. With this we have
\[
\glap'^\pseudo - S \preceq \frac{\sqrt{\epsilon}}{8} N \preceq  \frac{\sqrt{\epsilon}}{4} \glap'^\pseudo.
\]
and by rearranging we get $(1 - \frac{\sqrt{\epsilon}}{4}) \glap'^\pseudo  \preceq S \preceq \glap'^\pseudo.$ By combining this with the sparsification guarantee of $\glap'$ and the fact that $(1+x) (1+ 4 x)^{-1} \leq (1-x)$ for $x \in [0,1/2)$ we finally get 
\[
\left(1+ \sqrt{\epsilon}\right)^{-1} \glap^\pseudo \preceq 
\left(1-\frac{\sqrt{\epsilon}}{4}\right)
\left(1+\frac{\sqrt{\epsilon}}{4}\right)^{-1} \glap^\pseudo \preceq (1-\frac{\sqrt{\epsilon}}{4}) \glap'^\pseudo \preceq S
\]
and
\[
S \preceq \glap'^\pseudo \preceq \left(1-\frac{\sqrt{\epsilon}}{4}\right)^{-1} \glap^\pseudo \preceq \left(1+ \sqrt{\epsilon}\right) \glap^\pseudo.
\]

We now show that $S$ constructed in the above way can implicitly found quickly and stored in the desired space. Since S only uses the matrix $\glap'$ and the linear operator $N$ as inputs in our definition, we only need to show that these can be found in the claimed time and space, and that these can be computed with the desired success probabilities. By the standard machinery of graph sparsification, we are able to compute $\glap'$ in $\otilde(m)$ time which is an $\sqrt{\epsilon}$-spectral sparisifier of $\glap$ with probability $0.95$ and store it in $\otilde(n/\epsilon)$ space \cite{ST, SS, ZhuLO15, LeeS17}. Further, by the extensive literature on Laplacian linear system solvers \cite{ST,KoutisMP10,KoutisMP11,KelnerOSZ13,LeeS13,PengS14,CohenKPPR14,KyngS16}, we can construct the linear operator $N$ in $\otilde(n/\epsilon)$ time, store some representation of it in $\otilde(n/\epsilon)$ space, and apply this representation to a vector in $\otilde(n/\epsilon)$ time such that $N$ satisfies its spectral guarantee with probability $0.95$ (see \cite{PengS14} or \cite{KyngS16} for solvers that clearly meet the symmetry constraints, though other algorithms can be modified for this purpose as well). Therefore the described implicit form of $S$ can be computed and stored in the claimed time and space, and $S$ satisfies its spectral guarantee with probability $0.9$ by the union bound.

Finally, we need to show that we can apply $S$ to a vector in $\otilde(n/\epsilon)$ time. We first note that both $\glap'$ and $N$ can be applied to vectors in $\otilde(n/\epsilon)$ time. Since $S$ is $N$ times a matrix polynomial in $\glap' N$ of degree $z = \otilde(1)$, it follows that $S$ can be applied to a vector in $\otilde(n/\epsilon)$ time as well. Thus $S$ satisfies all of our desired properties.
\end{proof}

With Lemma~\ref{lemma:solver} with Theorem~\ref{thm:general} we now have everything we need to prove Theorem~\ref{thm:pinv}. 

\begin{proof}[Proof of Theorem~\ref{thm:pinv}]
Consider Algorithms \ref{alg:psketch} and \ref{alg:peval} below. We claim that these have the properties we desire in Theorem~ \ref{thm:pinv}. 

We begin with the error guarantee. Assume we are given a vector $\bv$ and wish to approximate $\bv^\top \glap^\pseudo \bv$. Note that our sketch computes $S$, a linear operator as in Lemma \ref{lemma:solver}, and then computes $\y = S \bv$. Once we have this, we take our $\epsilon$-spectral sketch $f$ computed as in Theorem \ref{thm:main}, and then compute $2\bv^\top \y - f(\y)$. Our goal is to show that this quantity is approximately equal to $\bv^\top \glap^\pseudo \bv$. By the guarantee from Lemma \ref{lemma:solver}, we have $(1+\sqrt{\epsilon})^{-1} \glap^\pseudo \preceq S \preceq (1+\sqrt{\epsilon}) \glap^\pseudo$ with probability $0.9$. Further by the guarantee from \ref{thm:main}, we have $(1 - \epsilon) \y^\top \glap \y \preceq f(\y) \preceq (1 + \epsilon) \y^\top \glap \y$ with probability $0.9$ as well. Thus, by the union bound, both $S$ and $f(\y)$ satisfy the requisite inequalities with probability $0.8$. If we define $\x$ to be such that $\glap \x = \bv$, by Theorem \ref{thm:general} (with $A = \glap$) , we have
\[
|2\bv^\top \y - f(\y) -  \bv^\top \glap^\pseudo \bv| \leq 4\epsilon q_{\glap, \bv} (\x) = 4 \epsilon \bv^\top \glap^\pseudo \bv.
\]
Thus, by appropriately scaling $\epsilon$, we obtain our guarantee. 

We now show the space bound and time complexity of constructing the sketch. Our sketch requires computing and storing two things: a representation of $S$ constructed via Lemma \ref{lemma:solver} and an $\epsilon$-spectral sketch for $\glap$ constructed by Theorem \ref{thm:main}. Now, by the space and construction bounds for each of these, computing our sketch clearly takes $\otilde(m)$ time and storing it costs $\otilde(n/\epsilon)$ space. 

Finally, we prove the running time per query. When evaluating our sketch, we must first apply $S$ to a vector. This can be done in $\otilde(n/\epsilon)$ by Lemma \ref{lemma:solver}. Once we have this vector, we must then compute a query of $f$ (which costs $\otilde(n/\epsilon)$ time by Theorem \ref{thm:main}) and then do $O(n)$ of additional work to compute a dot product. Thus, in total querying our sketch takes $\otilde(n/\epsilon)$ time, and all of our claims are proven.
\end{proof}

To conclude, we note that as in the previous section we can boost the constant probability of success of our algorithm to $1 - \frac{1}{\poly(n)}$ by storing a logarithmic number of copies of our sketch and taking the median  of the values obtained by querying each copy as our estimate.

\begin{algorithm}[ht]
	\caption{$\texttt{PseudoinverseSketch}$($G$, $\epsilon$)} \label{alg:psketch}
	\KwIn{A polynomially-bounded weighted graph $G = (V,E)$ with $n$ nodes; a tolerance parameter $\epsilon$}
	\KwOut{\textsf{data}($G$) data for a sketch of $\glap^\pseudo$}
	$\textsf{data}(G) \gets \emptyset$\;
	$S \gets$ a linear operator which acts as a pseudoinverse of some $\sqrt{\epsilon}$-spectral sparsifier for $\glap$ computed via Lemma \ref{lemma:solver} \;
	$f \gets$ an $\epsilon$-spectral sketch for $\glap$ computed via Theorem 1\;
	Add $f$ and $S$ to $\textsf{data}(G)$\;
	\Return{\textsf{data}($G$)}
\end{algorithm}

\begin{algorithm}[ht]
	\caption{$\texttt{PseudoinverseEval}$(\textsf{data}($G$),$\x$,$\epsilon$)} 
	\label{alg:peval}
	\KwIn{\textsf{data}($G$) data for a sketch of $\glap^\pseudo$; a query vector $\bv$}
	\KwOut{an approximation to $\bv^\top \glap^\pseudo  \bv$}
	$\y \gets S \bv $\; 
	\Return{$2\bv^\top \y - f(\y)$}
\end{algorithm}

%% file: allpairs.tex
\section{All Pairs Effective Resistances with Pseudoinverse Sparsifiers}
\label{sec:all_pairs}

In Section~\ref{sec:pseudo} we provided an algorithm that allows us to sketch the pseudoinverse of a Laplacian with a constant probability of failure for each vector in nearly linear time and using $\otilde(n/\epsilon)$ bits of data from the Laplacian. However, the query time for the function constructed is $\otilde(n/\epsilon)$: each time we run our algorithm we need to perform an approximate Laplacian system solve against a sparsified Laplacian. Although this runtime may be acceptable should we want to check a small number of vectors, if we wanted to compute effective resistances between every pair of nodes in a graph this would take $\otilde(n^3/\epsilon)$ work in total, which is slower than exactly computing the effective resistances via matrix multiplication. 

In this section, we demonstrate a way to preprocess our pseudoinverse sparsifier in $\otilde(n^2/\epsilon)$ time such that the resulting object can answer pseudoinverse queries on vectors with $k$ nonzero entries in $O(k^2)$ time. Since approximating effective resistances is equivalent to querying this data structure with vectors with two nonzero entries, we can thus use this structure to answer single effective resistance queries in constant time. This allows us, with the medians trick from previous sections, to estimate all pairs effective resistances of a graph in $\otilde(n^2/\epsilon)$, and ensure all of them are accurate within $1 \pm \epsilon$ with high probability. 

Recall that our pseudoinverse sparsifier was constructed in the following way: we stored a linear operator $S$ and our $\epsilon$-spectral sketch $f$ from Section \ref{sec:lap_sketch}. Whenever we received a query $\bv$, we found $\x = S \bv$, and we then computed $g(\x) = 2\bv^\top \x - f(\x)$. This was shown to be a $1 \pm \epsilon$ approximator for $\bv^\top \glap^\pseudo \bv$ with constant probability for any $\bv$. Our approach is as follows: we will take $f(\x)$ and convert it into an explicit matrix $M$ where $f(\x) = \x^\top M \x$, and then we will leverage the implicit sparsity of $S$ to efficiently compute matrix-matrix products with it. We will use this to build a matrix $C$, and we will prove that the quadratic forms of $C$ are approximately those of $\glap'^{\pseudo}$. 

We begin by showing the matrix $M$ described above exists, and can be found efficiently.

\begin{lemma}
\label{lemma:matrixize}
Let $\glap$ be a graph Laplacian for a polynomially-bounded weighted graph, and let $f$ be an $\epsilon$-spectral sketch for $\glap$ constructed as in Section \ref{sec:lap_sketch}. Then, there exists a matrix $M$ where $f(\x) = \x^\top M \x$ for all $\x$, and we can find $M$ in $\otilde(n^2)$ time. 
\end{lemma}

\begin{proof}
 Recall how $f$ was constructed. It was built by converting our input graph into a weighted combination of unweighted graphs $G_i$, and then applying expander partitioning to each $G_i$ separately to form $G_{ij}$s. We then sparsified the connected components of the $G_{ij}s$ separately. Indeed, we have the formula
\[
f(\x) = \sum_{i=1}^{\lfloor{\log M}\rfloor + 1} 2^{i-1}\sum_{j}\sum_k f_{ijk}([\x]_{H_{ijk}}),
\]
where $H_{ijk}$ is the set of nodes in the $k^{th}$ connected component in $G_{ij}$ and $[\x]_S$ is the vector $\x$ induced on the elements in the set $S$. Each $f_{ijk}$ was constructed in turn by scaling each corresponding component of the query vector by a multiple of the ones vector to form a vector $\y$, and then computing
\[
 \sum_{u \in H_{ijk}} y_u^2 \delta_u - \sum_{u \in S} \sum_{v \in V(H_{ijk})} y_u y_v -  \sum_{u \in L} \sum_{v \in S} y_u y_v - \sum_{u \in L} \frac{\delta^L_u}{\alpha} \sum_{v \in L} y_u y_v Y^u_v,
\]
where $\y$ is our scaled query and $\delta_i, \delta_u^L, \alpha, Y_v^u, S$, and $L$ are all randomly generated but fixed information dependent on $i$, $j$, and $k$. The four sums collectively sum $\otilde(|H_{ijk}|/\epsilon)$ different products $y_i y_j$: the first sum adds $|H_{ijk}|$ terms, and the last three sum one term for each of the $\otilde(|H_{ijk}|/\epsilon)$ edges our sketch stores. Thus, this equation represents a quadratic form $\y^\top L_{H_{ijk}}\y$, for some matrix $L_{H_{ijk}}$, and this matrix has $\otilde(|H_{ijk}|/\epsilon)$ nonzeroes. Now, the scaling by the ones vector is obviously a linear operator, and thus we can write $\y^\top L_{H_{ijk}}\y = [\x]_{H_{ijk}}^\top O_{ijk} L_{H_{ijk}} O_{ijk} [\x]_{H_{ijk}} = [\x]_{H_{ijk}}^\top M_{ijk} [\x]_{H_{ijk}}$ for some matrix $M_{ijk}$. Note that since $O_{ijk}$ simply removes a certain multiple of the ones vector from a query vector, it is a rank-one matrix plus the identity. Thus, we can compute its image upon a matrix in $\otilde(|H_{ijk}|^2)$ time, and so we can compute $M_{ijk}$ in $\otilde(|H_{ijk}|^2)$ time as well. As the connected components have no information crossing over from each other, we can combine these forms to get a larger quadratic form $\x^\top M_{ij} \x$. Therefore, we can conclude $f_{ij}$ has a corresponding matrix expression as
\[
f_{ij}(\x) = \x^\top M_{ij} \x,
\]
and that we can find $M_{ij}$ in $\otilde(n^2)$ time. Thus,
\[
f(\x) = \sum_{i=1}^{\lfloor{\log m}\rfloor + 1} 2^{i-1}\sum_{j} \x^\top M_{ij} \x = \x^\top \left(\sum_{i=1}^{\lfloor{\log m}\rfloor + 1} 2^{i-1}\sum_{j}  M_{ij}\right) \x = \x^\top M \x. 
\]
\end{proof}

In light of the above lemma, we can finally prove Theorem \ref{thm:aper}. We restate it for clarity.
\begin{algorithm}[t]
	\caption{$\texttt{AllPairsEffectiveResistances}$($G$, $\epsilon$)} \label{alg:aper}
	\KwIn{A polynomially-bounded weighted graph $G = (V,E)$ with $n$ nodes; a tolerance parameter $\epsilon$}
	\KwOut{\textsf{output}: $n^2$ values which estimate the effective resistance between every pair of nodes $(i,j)$. Any fixed pair is a $1 \pm \epsilon$ approximation with probability $0.8$. }
	\textsf{output} $\gets \emptyset$ \;
	$f \gets$ an $\epsilon$-spectral sketch for $\glap$ computed via Theorem \ref{thm:main} \;
	$M \gets$ matrix form of $f$ computed as in Lemma $\ref{lemma:matrixize}$ \;
	$S \gets$ a linear operator which acts as a pseudoinverse of some $\sqrt{\epsilon}$-spectral sparsifier for $\glap$ computed via Lemma \ref{lemma:solver} \;
	$Q \gets 2S^\top - S^\top MS$ \;
	\For {$i \in 1:n$}
		{\For {$j \in 1:n$}
			{ $\x \gets \textbf{0}$ \;
			$\x[i] = 1$, $\x[j] = -1$ \;
			Add $((i,j), \x^\top Q \x)$ to \textsf{output}\;
			}}

	\Return{\textsf{output}}
\end{algorithm}

\aper*

\begin{proof}
Consider Algorithm \ref{alg:aper} above. Given input $\glap$, first computes a matrix form $M$ of our sketch from Theorem \ref{thm:main}, and it then computes a linear operator $S$ by Lemma \ref{lemma:solver}. With these, it exactly computes $2S^\top - S^\top M S$ and returns it. We can generate implicit representations of both $M$ and $S$ in $\otilde(m)$ time, and since $M$ can be computed in $\otilde(n^2)$ time the cost of finding it explicitly can be rolled into the evaluation time. Note that if $\x' = S \bv$ then
\[
\bv^\top (2S^\top - S^\top M S) \bv = 2 \bv^\top \x' - \x'^\top M \x' = 2\bv^\top \x' - f(\x') ~,
\] and by the previous section this is a $(1 + \epsilon)$-multiplicative approximation of $\bv^\top \glap^\pseudo \bv$ with constant probability. Therefore $2S^\top - S^\top M S$ preserves each quadratic form of $\glap^\pseudo$ with constant probability. 

Note that $S$ can be applied to a vector in $\otilde(n/\epsilon)$ time by Lemma \ref{lemma:solver}, so both $MS$ and $S$ itself can be computed in $\otilde(n^2/\epsilon)$ time by computing the $n$ matrix-vector products of the rows of the corresponding matrix. Therefore, $S^\top M S$ can also be computed in $\otilde(n^2/\epsilon)$ time, and so we can obtain $2S^\top - S^\top M S$ exactly in $\otilde(n^2/\epsilon)$ time. Computing each of the quadratic forms corresponding to effective resistances can be done in $O(n^2)$ additional time, and so we can compute approximations to the all-pairs effective resistances in $G$ (which are each accurate with constant probability) in $\otilde(n^2/\epsilon)$ total time. 

By storing $\otilde(1)$ of these sketches, computing approximations to all of the pairwise effective resistances with each, and then returning the median value computed for each, we can boost the probability of success for every single pair to $1 - \frac{1}{poly(n)}$. Taking the union bound over all $O(n^2)$ pairs implies our result. 
\end{proof}